\documentclass[aps,twocolumn,groupedaddress,superscriptaddress,amsmath,amssymb,pra]{revtex4-2}
\usepackage{amsmath}
\usepackage{amssymb}
\usepackage{amsthm}
\usepackage{physics}
\usepackage{dsfont}
\usepackage{graphicx}
\usepackage{dcolumn}
\usepackage{bm}
\usepackage{bbold}
\usepackage[usenames,dvipsnames]{color}
\usepackage[colorlinks,citecolor=Blue,linkcolor=Red, urlcolor=Blue]{hyperref}

\hypersetup{breaklinks=true}

\usepackage{tikz}
\usepackage{pgfplots}
\usetikzlibrary{quantikz,angles,quotes}
\usetikzlibrary{decorations.pathmorphing}
\usetikzlibrary{arrows.meta}
\usepackage{braket}
\usepackage[normalem]{ulem}
\usepackage{subcaption}
\usepackage{float}

\newcommand{\C}{\mathbb{C}}

\newcommand{\R}{\mathbb{R}}
\newcommand{\Z}{\mathbb{Z}}

\newcommand{\sL}[2]{{\mathfrak{s l}(2, \mathbb{C})}}
\newcommand{\bigO}{\mathcal{O}}

\newcommand{\1}[1]{\mathds{1}}

\renewcommand{\ket}[1]{| #1 \rangle}
\renewcommand{\bra}[1]{\langle #1 |}
\renewcommand{\braket}[2]{\langle #1 | #2 \rangle}

\renewcommand{\ketbra}[2]{\ket{#1}\bra{#2}}

\newcommand{\mathbbm}[1]{\text{\usefont{U}{bbm}{m}{n}#1}}
\newcommand{\id}{\mathbbm{1}}

\def\qdots{\ \vdots\ }

\newtheorem{theorem}{Theorem}

\newtheorem{algorithm}[theorem]{Algorithm}
\newtheorem{problem}[theorem]{Problem}
\newtheorem{claim}[theorem]{Claim}

\newtheorem{definition}[theorem]{Definition}

\newtheorem{observation}[theorem]{Observation}

\begin{document}

\title{Quantum Signal Processing, Phase Extraction, and Proportional Sampling}

\author{Lorenzo Laneve}
\email{lorenzo.laneve@usi.ch}
\affiliation{Faculty of Informatics — Universit\`a della Svizzera Italiana, 6900 Lugano, Switzerland}


\begin{abstract}
\noindent Quantum Signal Processing (QSP) is a technique that can be used to implement a polynomial transformation $P(x)$ applied to the eigenvalues of a unitary $U$, essentially implementing the operation $P(U)$, provided that $P$ satisfies some conditions that are easy to satisfy. A rich class of previously known quantum algorithms were shown to be derived or reduced to this technique or one of its extensions. In this work, we show that QSP can be used to tackle a new problem, which we call \emph{phase extraction}, and that this can be used to provide quantum speed-up for \emph{proportional sampling}, a problem of interest in machine-learning applications and quantum state preparation. We show that, for certain sampling distributions, our algorithm provides an almost-quadratic speed-up over classical sampling procedures. Then we extend the result by constructing a sequence of algorithms that increasingly relax the dependence on the space of elements to sample.
\end{abstract}
\maketitle

\twocolumngrid 

\section{Introduction}
Quantum Signal Processing (QSP) is a novel technique to devise quantum algorithms developed by Low, Yoder, and Chuang~\cite{lowMethodologyResonantEquiangular2016, lowOptimalHamiltonianSimulation2017, lowHamiltonianSimulationQubitization2019} where, roughly speaking, one can transform some entry of a unitary matrix using a polynomial $P(x)$, i.e.,
\begin{align*}
    \bra{0} U \ket{0} = a \Rightarrow \bra{0} Q(U) \ket{0} = P(a),
\end{align*}
where $Q$ is a construction of a circuit that calls $U$ $n$~times if $n$ is the degree of $P$. Low et al.~\cite{lowMethodologyResonantEquiangular2016} proved that the class of polynomials realizable using their construction is quite general. Originally, this technique was used to tackle the Hamiltonian-simulation problem, where nearly-optimal complexities could be achieved for a huge class of Hamiltonians~\cite{lowHamiltonianSimulationUniform2017, lowHamiltonianSimulationQubitization2019}, even requiring a single copy of the initial state (so-called fully coherent simulation~\cite{martynEfficientFullyCoherentQuantum2022}). The QSP construction was extended to two similar techniques: the first is called Quantum Eigenvalue Transform (QET), where the polynomial transformation $P(x)$ induced by the QSP could be applied to all the eigenvalues of an arbitrary block-encoded square matrix $A$ (essentially implementing $P(A)$) by using, again, $n$ calls to $U$~\cite{lowHamiltonianSimulationUniform2017, haahProductDecompositionPeriodic2019}. The second, introduced by Gilyen~\cite{gilyenQuantumSingularValue2019a, gilyenQuantumSingularValue2019c}, was about transforming the \emph{singular values} of a block-encoded matrix $A$. The latter extension, which is the most general one, was also shown to reproduce most of famous quantum algorithms present in the literature, such as Grover's and Shor's algorithms~\cite{groverFastQuantumMechanical1996a, shorPolynomialTimeAlgorithmsPrime1997}, but also the HHL algorithm for solving linear systems~\cite{harrowQuantumAlgorithmSolving2009}, general amplitude-amplification schemes and the phase-estimation procedure~\cite{nielsenQuantumComputationQuantum2010b}. Moreover, works by Haah~\cite{haahProductDecompositionPeriodic2019} and Chao et al.~\cite{chaoFindingAnglesQuantum2020}, which had the goal to overcome some technical problems of the QSP (namely, classical computation of the phase sequence needed in the construction), introduced novel formalisms that are relevant for the present work (see Appendix~\ref{apx:haah-construction} for a brief introduction).

Starting from Section~\ref{sec:phase-extraction-problem}, we present a different problem, called \emph{phase extraction}, in which one basically wants to construct a (block-encoded) Hermitian matrix $H$ (or some real function of it) using copies of $U = e^{i \pi H}$ and its inverse or, in other words, to transform eigenphases to (real) amplitudes. We show a general construction, using QET and Fourier series, to approximate any (sufficiently smooth) function of the phases with Laurent polynomials that can be directly implemented with Haah's formalism~\cite{haahProductDecompositionPeriodic2019}.

In Section~\ref{sec:prop-sampling} we proceed by showing a natural application of the phase extraction problem to \emph{proportional sampling}, a problem of interest in machine learning where one wants to sample some element $x$ of a given set with a probability that is proportional to the value given by some oracle $c(x)$. We devise an algorithm using the phase-extraction procedure to construct a quantum state where the amplitude of each element is the square root of its sampling probability, so that a measurement in the computational basis will induce our desired probability distribution. If we omit the measurement, we can consider this algorithm as a way to prepare a quantum state with a specific shape defined through the oracle, and this can be relevant in the topic of quantum state preparation. We close this section by proving that, in some instances, the algorithm achieves an almost quadratic speed-up.

As a final step, in Section~\ref{sec:inductive-smoothening} we show how further speed-up can be achieved `for free' if we carefully choose the function to approximate in the phase-extraction subroutine. In its final form, our algorithm does not (directly) depend on the dimensionality $N$ of our Hilbert space. In contrast to previous algorithms for quantum state preparation~\cite{knillApproximationQuantumCircuits1995, mottonenTransformationQuantumStates2004a, araujoDivideandconquerAlgorithmQuantum2021, zhangQuantumStatePreparation2022}, which require $\Theta(N)$ depth or width, our proposed solution may take as low as $\bigO(\log N)$ depth for some instances, and always $\bigO(\log N)$ qubits.
\section{Phase Extraction}
\label{sec:phase-extraction-problem}

In this section we define a new problem, which we call \emph{phase extraction}. In some sense, this problem can be seen as the inverse of the simulation problem: while, in order to simulate an Hamiltonian, we need to construct the complex exponential of the given matrix, in the phase extraction problem we essentially want to extract its complex logarithm.

\begin{problem}[Phase Extraction]
    \label{def:phase-extraction}
    Let $H$ be an Hermitian matrix satisfying $||H|| < 1$. Given a controlled version of the unitaries $U = e^{i \pi H}$ and $U^\dag = e^{-i \pi H}$ as oracles and $\varepsilon > 0$, construct a quantum circuit $C_U$ such that:
    $$||\bra{0}_A C_U \ket{0}_A - H|| \le \varepsilon,$$
    where $A$ is the subsystem containing ancilla registers.
\end{problem}
\noindent In other words, if we initially have a quantum circuit acting on its eigenbasis $\{ \ket{\phi_k} \}_k$ as
\begin{align*}
    U \ket{\phi_k} = e^{i \pi \varphi_k} \ket{\phi_k}
\end{align*}
for $\varphi_k \in [0, 1)$, we would like $C_U$ to act on the same eigenbasis as
\begin{align*}
    C_U \ket{0} \ket{\phi_k} = \varphi_k \ket{0} \ket{\phi_k} + \ket{1} \ldots \ .
\end{align*}
Thus, $C_U$ will contain a so-called \emph{block encoding} of the matrix $H$. The reader might argue that the stated problem is ill-defined: a global phase on $U$ would change our desired output, but not the input. However, it is important to note that the unitaries $U, U^\dag$ have to be given in their \emph{controlled} versions, where a global phase on $U$ would be seen as a controlled phase kickback. Moreover, it is worth remarking that this problem is substantially different from the well-known \emph{phase-estimation} problem, where one wants to achieve classical information about~$\varphi_k$, given $\ket{\phi_k}$.

Now we would like to use the quantum eigenvalue transform to apply a polynomial transformation of the eigenvalues of $U$, using the oracles we have. In particular, Haah's work tells us that any Laurent polynomial transformation $F: U(1) \rightarrow SU(2)$ of degree $n$ can be implemented using the QSP construction with only $\bigO(n)$ calls to $U, U^\dag$, provided that $F$ has definite parity (see Appendix~\ref{apx:haah-construction}). It is sufficient for us to have a desired transformation $f: U(1) \rightarrow [-1, 1]$ on one entry of this matrix (from now on we will assume without loss of generality that this entry is the top-left one, i.e., $f(z) = \bra{0} F(z) \ket{0}$).

Therefore, all we need now is to design a class of polynomials uniformly approximating a function of the phase. An interesting observation is that Laurent polynomials on the unit circle can be constructed using Fourier series.

\begin{observation}
    \label{thm:phase-extraction-poly}
    Let $\phi : \R \mapsto \R$ be a $2\pi$-periodic function whose Fourier series uniformly converges, and let $f : U(1) \mapsto \R$ be such that, for any real $x$,
    \begin{align*}
        f(e^{i x}) = \phi(x)
    \end{align*}
    holds. Then, there is a sequence of complex (Laurent) polynomials $P_n : \C \mapsto \C$ which uniformly converges to $f$ on the unit circle.
\end{observation}
This theorem tells us that we can construct a sequence of Laurent polynomials approximating any function of the eigenphase, and this will inherit all the strong convergence properties of Fourier sequences in our domain of interest. From now on, we use $\bar{\phi}(z)$ to denote the Laurent polynomial (or Laurent series) that satisfies $\bar{\phi}(e^{ix}) = \phi(x)$ for any real $x$. Notice that such function is unique since polynomials are fully defined by their behaviour on an infinite set.
\begin{proof}
    If $\phi(x) = \sum_{k \in \Z} c_k e^{i k x}$ is the Fourier series of $\phi(x)$, then
    \begin{align*}
        \bar{\phi}(z) = \sum_{k \in \Z} c_k z^k
    \end{align*}
    and this function is equal to $f$ on the unit circle. If $P_n(x) = \sum_{k = -n}^n c_k e^{ikx}$, then it is sufficient to see:
    \begin{align*}
        ||P_n - \phi||_{\R} = ||\bar{P}_n - \bar{\phi}||_{U(1)} = ||\bar{P}_n - f||_{U(1)},
    \end{align*}
    and the first norm tends to zero as $n \rightarrow \infty$.
\end{proof}
This result can be used to solve an extension of Problem~\ref{def:phase-extraction}. Indeed, we can compute a block encoding of $\phi(H)$ for an arbitrary real function $\phi$, provided it is sufficiently `well-behaved'.

Let us solve Problem~\ref{def:phase-extraction} using this technique: we would like that, for $\theta \in (-\pi, \pi]$, the eigenvalue $e^{i \theta}$ is mapped to $\theta/\pi$. Realistically, we cannot approximate this function in the whole interval, since there is a discontinuity in $\theta = \pm \pi$ and, because of this, the Fourier sequence could take too long to converge, or it could not uniformly converge at all. Therefore, we give up a small portion of the interval: for some small $\delta > 0$, we design a function $\phi_\delta(x)$ that is equal to $\phi(x) = x/\pi$ for every $x \in (-\pi + \delta, \pi - \delta)$, while in the neighbourhoods of $\pm \pi$ we replace its derivative with something continuous and piecewise linear that preserves the $2\pi$-periodicity of $\phi_\delta(x)$ (see Figure~\ref{fig:function-approx}). This implies that $\phi_\delta \in C^1$ by construction and, intuitively, the associated Fourier sum will converge faster. As a consequence, we will have to restrict Problem~\ref{def:phase-extraction} to instances with $||H|| \le 1 - \delta$.

\begin{figure}
    \centering
    \begin{subfigure}[b]{\columnwidth}
         \centering
         \begin{tikzpicture}
\begin{axis}[
    width=250pt,height=100pt,
    xmin=-4,xmax=4,
    ymin=-1.2,ymax=1.2,
    samples=50,
    xtick={-pi, -pi/2, 0, pi/2, pi},
    xticklabels={$-\pi$, $-\frac{\pi}{2}$, $0$, $\frac{\pi}{2}$, $\pi$},
    grid style={line width=.1pt, draw=gray!10},
    axis line style={latex-latex}]
    
    \addplot[blue, ultra thick, domain=-pi:pi] (x, x/pi);
    \addplot[blue, ultra thick, domain=-4:-pi] (x, x/pi + 2);
    \addplot[blue, ultra thick, domain=pi:4] (x, x/pi - 2);

    \addplot[mark=*] coordinates {(-pi,1)};
    \addplot[mark=o] coordinates {(-pi,-1)};
    \addplot[mark=*] coordinates {(pi,1)};
    \addplot[mark=o] coordinates {(pi,-1)};

    \draw [dashed] (axis cs:{-pi},-2) -- (axis cs:{-pi},2);
    \draw [dashed] (axis cs:{pi},-2) -- (axis cs:{pi},2);

    \node at (axis cs:-1.5,0.6) {$\phi(x)$};
  
\end{axis}
\end{tikzpicture}
    \end{subfigure}
    \hfill
    \begin{subfigure}[b]{\columnwidth}
        \centering
        \def\dlt{0.4}
\begin{tikzpicture}
\begin{axis}[
    width=250pt,height=110pt,
    xmin=-4,xmax=4,
    ymin=-1.2,ymax=1.45,
    samples=50,
    xtick={-pi, -pi/2, 0, pi/2, pi},
    xticklabels={$-\pi$, $-\frac{\pi}{2}$, $0$, $\frac{\pi}{2}$, $\pi$},
    grid style={line width=.1pt, draw=gray!10}]
    
    \addplot[red, ultra thick, domain=-pi+\dlt:pi-\dlt] (x, x/pi);
    \addplot[red, ultra thick, domain=-4:-pi-\dlt] (x, x/pi + 2);
    \addplot[red, ultra thick, domain=pi+\dlt:4] (x, x/pi - 2);

    \addplot[red, ultra thick, domain=-pi:-pi+\dlt] (x, x/pi + x*x/\dlt/\dlt + 2*pi*x/\dlt/\dlt - 2*x/\dlt + 1 + pi*pi/\dlt/\dlt - 2*pi/\dlt);
    \addplot[red, ultra thick, domain=-pi-\dlt:-pi] (x, 1 - 2*pi/\dlt - 2*x/\dlt + x/pi - pi*pi/\dlt/\dlt - pi*x/\dlt/\dlt - pi*x/\dlt/\dlt - x*x/\dlt/\dlt);
    \addplot[red, ultra thick, domain=pi-\dlt:pi] (x, -1 + 2*pi/\dlt - 2*x/\dlt + x/pi - pi*pi/\dlt/\dlt + pi*x/\dlt/\dlt + pi*x/\dlt/\dlt - x*x/\dlt/\dlt);
    \addplot[red, ultra thick, domain=pi:pi+\dlt] (x, -1 + 2*pi/\dlt - 2*x/\dlt + x/pi + pi*pi/\dlt/\dlt - pi*x/\dlt/\dlt - pi*x/\dlt/\dlt + x*x/\dlt/\dlt);

    \draw [dashed] (axis cs:{pi-\dlt},-2) -- (axis cs:{pi-\dlt},2);
    \draw [dashed] (axis cs:{pi+\dlt},-2) -- (axis cs:{pi+\dlt},2);
    \draw [dashed] (axis cs:{-pi-\dlt},-2) -- (axis cs:{-pi-\dlt},2);
    \draw [dashed] (axis cs:{-pi+\dlt},-2) -- (axis cs:{-pi+\dlt},2);

    \draw[|<->|] (axis cs:{-pi-\dlt},1) -- node[above] {$2\delta$} (axis cs:{-pi+\dlt},1);
    \draw[|<->|] (axis cs:{pi-\dlt},1) -- node[above] {$2\delta$} (axis cs:{pi+\dlt},1);

    \node at (axis cs:-1.5,0.85) {$\phi_\delta(x)$};
    
\end{axis}
\end{tikzpicture}
\undef\dlt
    \end{subfigure}
    \caption{Construction of $\phi_\delta(x)$ from $\phi(x)$. The two functions are identical except in the $\delta$-neighbourhood of the odd multiples of $\pi$. We have that $\phi_\delta \in C^1$ for any $\delta > 0$, so its Fourier series will converge faster than the series of $\phi(x)$.}
    \label{fig:function-approx}
\end{figure}
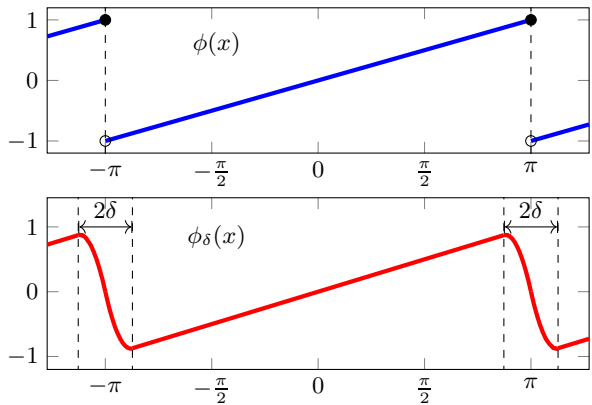

\begin{theorem}
    \label{thm:phase-extraction-jackson-rate}
    The function $\bar{\phi}_\delta : U(1) \rightarrow \R$ defined as
    \begin{align*}
        \bar{\phi}_\delta(e^{ix}) = \phi_\delta(x)
    \end{align*}
    for every $x \in \R$ can be $\epsilon$-approximated on the unit circle using a polynomial of degree
    $$d = \Tilde{\bigO}\left(\frac{1}{\delta} \sqrt{\frac{1}{\epsilon}}\right).$$
\end{theorem}
\begin{proof}
    Let $S_{\delta, d}(x) = \sum_{k = -d}^{d} c_k e^{ikx}$ be the Fourier sum of $\phi_\delta$ up to terms of degree $d$. By a result of Jackson~\cite[pp.\ 20--25]{jacksonTheoryApproximation1930a} we know that, since $\phi_\delta \in C^1$ and $\phi'_\delta$ is $2/\delta^2$-Lipschitz, the approximation error of the $d$-th degree Fourier sum is bounded by
    \begin{align*}
        ||S_{\delta, d} - \phi_\delta||_{\R} \le K \frac{\log d}{\delta^2 d^2} = K \frac{1}{\delta^2 d^{2-o(1)}}
    \end{align*}
    for some absolute constant $K$. The above is bounded by $\epsilon$ when $d^{2-o(1)} = \bigO\left(\frac{1}{\epsilon \delta^2}\right)$ or, in other words
    \begin{align*}
        d = \Tilde{\bigO}\left( \frac{1}{\delta} \sqrt{\frac{1}{\epsilon}} \right).
    \end{align*}
    This concludes the proof since we can replace $e^{ix} = z$ in $S_{\delta, d}$ to obtain the sequence of Laurent polynomials $\bar{S}_{\delta, d}$ uniformly converging to $\bar{\phi}_\delta$ on the unit circle with the same rate.
\end{proof}
It is interesting to point out that, if one only cares about a constant $\delta$ (e.g., $\delta = \pi/2$ so we get our good approximation only on the right semicircle), one could increase the smoothness of $\phi_\delta$ to get a better dependence from $1/\epsilon$. The Lipschitz constant increases, but it would only depend on $\delta$. We investigate this improvement in Section~\ref{sec:inductive-smoothening}.

Now we have a Laurent polynomial approximating the function $\bar{\phi}_\delta$ that we would like to achieve. We need to take care of one last problem: remember that Haah's construction requires the Laurent polynomial to be either even or odd. Keep in mind that $\bar{\phi}_\delta(z)$ does not have definite parity (and so do not its approximating polynomials), even though $\phi_\delta$ is odd.

A simple workaround is to split $\bar{S}_{\delta, d}(z)$ into even and odd polynomials $\bar{S}^0_{\delta, d}(z), \bar{S}^1_{\delta, d}(z)$, implement them separately, and then add them up using a simple block encoding (see Figure~\ref{fig:block-encoding-sum-circuit}). The phase extraction function $\phi(x)$ can be split into $\phi_0, \phi_1$ such that $\phi_0(x) = \phi_0(x + \pi)$ and $\phi_1(x) = -\phi_1(x + \pi)$ (the reader can check that the Laurent polynomials associated to their Fourier series are even and odd, respectively). Moreover, both of them are bounded by~$1/2$ in absolute value everywhere, and the sum of their Laurent polynomials gives exactly the Laurent polynomial of $\phi$ (see Figure~\ref{fig:function-approx-even-odd}).

\begin{figure}
    \centering
    \begin{subfigure}[b]{\columnwidth}
         \centering
         \begin{tikzpicture}
\begin{axis}[
    width=250pt,height=100pt,
    xmin=-4,xmax=4,
    ymin=-0.6,ymax=0.6,
    samples=50,
    xtick={-pi, -pi/2, 0, pi/2, pi},
    xticklabels={$-\pi$, $-\frac{\pi}{2}$, $0$, $\frac{\pi}{2}$, $\pi$},
    grid style={line width=.1pt, draw=gray!10},
    axis line style={latex-latex}]
    
    \addplot[green, ultra thick, domain=-pi:0] (x, x/pi + 1/2);
    \addplot[green, ultra thick, domain=0:pi] (x, x/pi - 1/2);
    \addplot[green, ultra thick, domain=pi:4] (x, x/pi - 3/2);
    \addplot[green, ultra thick, domain=-4:-pi] (x, x/pi + 3/2);

    \addplot[mark=*] coordinates {(-pi,1/2)};
    \addplot[mark=o] coordinates {(-pi,-1/2)};
    \addplot[mark=*] coordinates {(0,1/2)};
    \addplot[mark=o] coordinates {(0,-1/2)};
    \addplot[mark=*] coordinates {(pi,1/2)};
    \addplot[mark=o] coordinates {(pi,-1/2)};

    \draw [dashed] (axis cs:{-pi},-2) -- (axis cs:{-pi},2);
    \draw [dashed] (axis cs:0,-2) -- (axis cs:0,2);
    \draw [dashed] (axis cs:{pi},-2) -- (axis cs:{pi},2);

    \node at (axis cs:-1.5,0.35) {$\phi_0(x)$};
  
\end{axis}
\end{tikzpicture}
    \end{subfigure}
    \hfill
    \begin{subfigure}[b]{\columnwidth}
        \centering
        \begin{tikzpicture}
\begin{axis}[
    width=250pt,height=100pt,
    xmin=-4,xmax=4,
    ymin=-0.6,ymax=0.6,
    samples=50,
    xtick={-pi, -pi/2, 0, pi/2, pi},
    xticklabels={$-\pi$, $-\frac{\pi}{2}$, $0$, $\frac{\pi}{2}$, $\pi$},
    grid style={line width=.1pt, draw=gray!10},
    axis line style={latex-latex}]
    
    \addplot[yellow, ultra thick, domain=-pi:0] (x, -1/2);
    \addplot[yellow, ultra thick, domain=0:pi] (x, 1/2);
    \addplot[yellow, ultra thick, domain=pi:4] (x, -1/2);
    \addplot[yellow, ultra thick, domain=-4:-pi] (x, 1/2);

    \addplot[mark=*] coordinates {(-pi,1/2)};
    \addplot[mark=o] coordinates {(-pi,-1/2)};
    \addplot[mark=o] coordinates {(0,1/2)};
    \addplot[mark=*] coordinates {(0,-1/2)};
    \addplot[mark=*] coordinates {(pi,1/2)};
    \addplot[mark=o] coordinates {(pi,-1/2)};

    \draw [dashed] (axis cs:{-pi},-2) -- (axis cs:{-pi},2);
    \draw [dashed] (axis cs:0,-2) -- (axis cs:0,2);
    \draw [dashed] (axis cs:{pi},-2) -- (axis cs:{pi},2);

    \node at (axis cs:-1.5,0.35) {$\phi_1(x)$};
  
\end{axis}
\end{tikzpicture}
    \end{subfigure}
    \caption{Plots of $\phi^0(x)$ and $\phi^1(x)$. One can see that their sum will give exactly $\phi(x)$. Since $\phi(x + \pi) = \bar{\phi}(e^{i(x + \pi)}) = \bar{\phi}(-e^{ix})$, this determines the parity of the approximating Laurent polynomials.}
    \label{fig:function-approx-even-odd}
\end{figure}
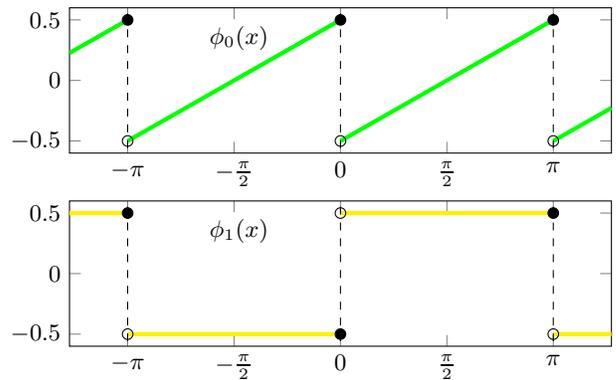

In the end, this procedure has a certain failure probability due to non-unitarity of the matrix we are extracting, which can be mitigated by classical repetition or amplitude amplification techniques~\cite{brassardQuantumAmplitudeAmplification2002, berryExponentialImprovementPrecision2014a, berrySimulatingHamiltonianDynamics2015}. It is now important to estimate the initial probability of measuring $\ket{00}$ (i.e., our success probability), in order to bound the multiplicative factor that the amplifying procedure gives to our overall complexity. However, this strictly depends on the unitary $U$ whose phases we want to extract: for example, if $U$ has all the eigenphases close to $0$, then all of them will be mapped to very small amplitudes by $\bar{\phi}_\delta(z)$, giving us a low probability of success, which is harder to amplify. In the next section we will see a concrete example.

\begin{figure}
    \centering
    \input{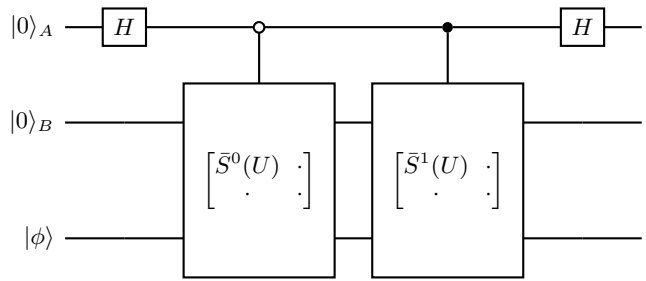}
    \caption{Implementation of the block encoding for summing two matrices. The two big gates represent a quantum eigenvalue transform on the unitary $U$ implementing the polynomials $\bar S^0, \bar S^1$. If $\ket{00}$ is measured on the two control qubits, then $\ket{\phi}$ will be transformed by the matrix $\bar S^0(U) + \bar S^1(U) \equiv \bar S(U)$, implementing the original Fourier approximation.}
    \label{fig:block-encoding-sum-circuit}
\end{figure}
\section{Application to Proportional Sampling}
\label{sec:prop-sampling}

\begin{figure*}
    \centering
    \input{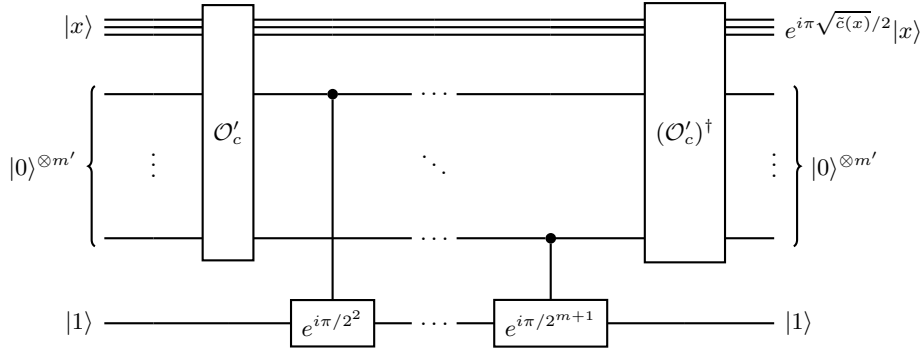}
    \caption{Implementation of a phase oracle for $\sqrt{\Tilde{c}(x)}$ using one call to $\bigO'_c, (\bigO'_c)^\dag$. The total phase applied to the quantum state will be $e^{i \pi \sqrt{\Tilde{c}(x)} / 2}$ where $|\sqrt{\Tilde{c}(x)} - \sqrt{c(x)}| \le 2^{-m}$. The two ancilla registers can then be discarded.}
    \label{fig:phase-oracle}
\end{figure*}

\noindent We consider the following problem.
\begin{problem}
    \label{def:prop-sampling}
    Consider a set $X$ of elements. We are given a function $c(x) : X \mapsto [0, 1)$ as an oracle, i.e.,
    \begin{align*}
        \bigO_c \ket{x} \ket{0}^{\otimes m} = \ket{x} \ket{c(x)} \ ,
    \end{align*}
    where $\ket{c(x)}$ contains the $m$ bits of $c(x)$ after the decimal point. Sample a value $x \in X$ such that $P(x) \propto c(x)$ up to some error $\epsilon > 0$. In other words:
    \begin{align*}
        \left| P(x) - \frac{c(x)}{\sum_{x \in X} c(x)} \right| \le \epsilon \ .
    \end{align*}
\end{problem}
For ease of exposition, we assume that the elements of $X$ are encoded as integers, i.e., $X = \{ 0, \ldots, N - 1 \}$, where $N \le 2^n \le 2N$. Moreover, we assume to know the \emph{average oracle value} $\bar{c} := \frac{1}{N} \sum_{x \in X} c(x)$.

In pratical applications, one usually wants to bound the so-called \emph{total variation distance} between the target and implemented probability distributions (see Section 4.1 of~\cite{levinMarkovChainsMixing2017} for a nice introduction to this metric). In order to bound this distance by $\epsilon$ it is sufficient to solve Problem~\ref{def:prop-sampling} within $\epsilon/N$ error. Quantum speed-up for this stronger version of the problem can be achieved using the improvements proposed in Section~\ref{sec:inductive-smoothening}.

\subsection{Constructing a phase oracle}
Since we want $c(x)$ as a probability, the idea is to extract the phases using the transformation $\bar{\phi}(z)$ designed in Section~\ref{sec:phase-extraction-problem} on the unitary
\begin{align*}
    U \ket{x} = e^{i \pi \sqrt{c(x)} / 2} \ket{x} \ .
\end{align*}
The factor $\pi/2$ will be clear later. Of course we cannot reproduce an arbitrary square root with an infinite number of digits as our eigenphases, but here we show a simple construction for the following unitary
\begin{align*}
    U' \ket{x} = e^{i \pi \sqrt{\Tilde{c}(x)} / 2} \ket{x} \ ,
\end{align*}
where $\sqrt{\Tilde{c}(x)}$ is the $m'$-bit truncation of $\sqrt{c(x)}$. Using the bit oracle $\bigO_c$ given in Problem~\ref{def:prop-sampling}, we can construct a bit oracle $\bigO'_c$ behaving as
\begin{align*}
    \bigO'_c \ket{x} \ket{0}^{\otimes m'} = \ket{x} \ket{\sqrt{\Tilde{c}(x)}} \ ,
\end{align*}
by decorating one call to the original oracle with a square root algorithm of $m'$-bit precision.

Now, using a standard construction (depicted in Figure~\ref{fig:phase-oracle}), we can implement $U'$ using $\bigO'_c$ and $(\bigO'_c)^\dag$. Only two copies of the original oracles $\bigO_c, \bigO^\dag_c$ are needed to construct $U'$.

\subsection{Bounding the approximation error}
By transforming $U'$ using $\bar{\phi}_\delta(z)$ we obtain an Hermitian block-encoded matrix $H_c$ acting as
\begin{align}
    H_c \ket{x} & = \bar\phi_\delta(e^{i \pi \sqrt{\Tilde{c}(x)} / 2}) \ket{x}\nonumber \\
    & = \phi_\delta \left(\frac{\pi}{2} \sqrt{\Tilde{c}(x)} \right) \ket{x} = \frac{1}{2} \sqrt{\Tilde{c}(x)} \ket{x} \label{eq:prop-sampling-ideal-transform}
\end{align}
provided that $|\frac{\pi}{2} \sqrt{\Tilde{c}(x)}| \le \pi - \delta$. It should now be evident why we added the $\pi/2$ factor, as we can fix $\delta$ as high as $\pi/2 = \Theta(1)$. If we applied this ideal transformation to the uniform superposition, we would obtain
\begin{align*}
    H_c \ket{+}^{\otimes n} = \frac{1}{2 \sqrt{N}} \sum_{x \in X} \sqrt{\Tilde{c}(x)} \ket{x}
\end{align*}
We have two sources of error on the final sampling probabilities: the first one, $|c(x) - \Tilde{c}(x)|$, is given by the $m'$-bit truncation and can be bounded easily
\begin{align*}
    |c(x) - \Tilde{c}(x)| & = |\sqrt{c(x)} + \sqrt{\Tilde{c}(x)}| \cdot |\sqrt{c(x)} - \sqrt{\Tilde{c}(x)}| \\
    & \le 2 \cdot \frac{1}{2^{m'}} = \frac{1}{2^{m'-1}} \stackrel{!}{\le} \frac{\epsilon'}{2} \ ,
\end{align*}
where the last inequality holds for $m' = \bigO(\log \epsilon')$. The second noise comes from the approximation of $\phi_\delta$ by $S_{\delta, d}$: if the approximation is up to $\epsilon'/16$ then
\begin{align*}
    \left| S^2_{\delta, d}\left(\frac{\pi}{2} \sqrt{c} \right) - \frac{1}{4} c \right| & = \left| S^2_{\delta, d}\left(\frac{\pi}{2} \sqrt{c} \right) - \phi^2_\delta\left(\frac{\pi}{2} \sqrt{c} \right) \right| \\
    & \le ||S_{\delta, d} + \phi_\delta||_{\R} \cdot ||S_{\delta, d} - \phi_\delta||_{\R} \\
    & \le 2 \cdot \frac{\epsilon'}{16} = \frac{\epsilon'}{8}
\end{align*}
for any $c \in [0, 1)$ (the $1/4$ factor comes from the $1/2$ factor in Eq.~(\ref{eq:prop-sampling-ideal-transform})). This is guaranteed by Theorem~\ref{thm:phase-extraction-jackson-rate} if we take
\begin{align*}
    d = \Tilde{\bigO}\left( \frac{1}{\delta} \sqrt{\frac{8}{\epsilon'}} \right) = \Tilde{\bigO}\left( \sqrt{\frac{1}{\epsilon'}} \right)
\end{align*}
as $\delta$ was already fixed to be constant. Therefore, the distance between the ideal oracle $c(x)$ and our implementation $s(x) := 4 S^2_{\delta, d}(\frac{\pi}{2} \sqrt{\Tilde{c}(x)})$ is
\begin{align*}
    \left|s(x) - c(x)\right| & \le \left|s(x) - \Tilde{c}(x)\right| + |\Tilde{c}(x) - c(x)| \\
    & \le 4 \frac{\epsilon'}{8} + \frac{\epsilon'}{2} \le \epsilon' \ .
\end{align*}
This is needed to bound the error induced in the sampling probabilities, but it is not sufficient, as this error bound will now be amplified by the amplitude amplification scheme.

\subsection{Amplifying the state}
The state $H_c \ket{+}^{\otimes n}$ is sub-normalized, because $H_c$ is not a unitary transformation, but only a block-encoded Hermitian matrix. By looking again at Figure~\ref{fig:block-encoding-sum-circuit}, we can see that we have two control qubits, $A$ and $B$. We measure both these qubits in the computational basis and, if we measure $\ket{00}$, then we picked the correct block. The probability of this happening is:
\begin{align*}
    || H_c \ket{+}^{\otimes n} ||^2 & = \frac{1}{N} \sum_{x \in X} S^2_{\delta, d} \left(\frac{\pi}{2} \sqrt{\Tilde{c}(x)} \right) \\
    & \ge \frac{1}{4 N} \sum_{x \in X} (c(x) - \epsilon') =: \frac{1}{4}(\bar{c} - \epsilon')\ ,
\end{align*}
and we take $\epsilon' = \frac{\bar{c} \epsilon}{2} \le \frac{\bar{c}}{2}$. Since an upper bound $\frac{1}{4}(\bar{c} + \epsilon')$ can be derived in a similar fashion, the initial amplitude is $\Theta(\sqrt{\bar{c}})$ and we can employ an \emph{amplitude amplification} procedure (see Appendix~\ref{apx:amplitude-amplification}) to normalize the state with $\bigO\left(\frac{1}{\sqrt{\bar{c}}}\right)$ repetitions. Using $c^*, s^*$ as shorthands for $\sum_x c(x), \sum_x s(x)$, we get
\begin{align*}
    \left| \frac{c(x)}{c^*} - \frac{s(x)}{s^*}\right| & = \frac{1}{s^* c^*} \left|c(x) s^* - s(x) c^* \right| \\
    & = \frac{s(x)}{s^* c^*} \left| s^* - c^* \right| + \frac{1}{c^*} |c(x) - s(x)| \\
    & \le \frac{1}{c^*} \left| s^* - c^* \right| + \frac{1}{c^*} |c(x) - s(x)| \\
    & \le \frac{N \epsilon'}{c^*} + \frac{\epsilon'}{c^*} = \frac{\epsilon'}{\bar{c}} + \frac{\epsilon'}{N \bar{c}} \\
    & \le \frac{2 \epsilon'}{\bar{c}}  \stackrel{!}{\le} \epsilon \ ,
\end{align*}
and one can see that our choice of $\epsilon'$ always satisfies the last inequality. To sum up, we obtained the following algorithm:
\begin{algorithm}[Proportional Sampling by QSP]
\label{alg:prop-sampling-qsp}
Let $\bar{c} = \frac{1}{N} \sum_{x \in X} c(x)$ be the average oracle value, and fix $\delta = \frac{\pi}{2}$.
\begin{enumerate}
        \item Use Quantum Eigenvalue Transform on the eigenvalues of $U'$ using $\bar{S}_{\delta, d}(z)$, obtaining a block-encoding of $\bar{S}_{\delta, d}(U')$. This requires $d$ calls to $U'$ or, equivalently, $2d$ calls to $\bigO_c$.

        \item Compute the state $\bar{S}_{\delta, d}(U') \ket{+}^{\otimes n}$.

        \item Use OAA to amplify the above state. This requires $\bigO(\frac{1}{\sqrt{\bar{c}}})$ copies of $\bar{S}_{\delta, d}(U')$.
    \end{enumerate}
    The number of total calls to $\bigO_c$ will be
    \begin{align*}
        \bigO \left( \frac{1}{\sqrt{\bar{c}}} \cdot d \right) = \Tilde{\bigO} \left( \frac{1}{\bar{c}} \sqrt{\frac{1}{\epsilon}} \right)
    \end{align*}
    to achieve error up to $\epsilon$ with high probability.
\end{algorithm}
We remark that Algorithm~\ref{alg:prop-sampling-qsp} is a so-called \emph{Las Vegas} algorithm: whenever it fails (i.e., we pick the wrong block of the encoding), we immediately know, because we measure $\neq 00$ on the control qubits. Therefore, if this check fails, we repeat the whole algorithm, and the number of repetitions is constant in expectation, as the probability of success is lower bounded by a constant after the amplification scheme.

\subsection{Separation proof}
In this subsection, we show that Algorithm~\ref{alg:prop-sampling-qsp} gives actual speed-up over any classical algorithm. Consider the following instance: we assume $N$ to be even and divide $\{ 0, \ldots, N - 1 \}$ into two equally sized sets $A, B$, where
\begin{align*}
    c(x) =
    \begin{cases}
        \frac{1}{4} & x \in A \\
        \frac{1}{8} & x \in B
    \end{cases}
\end{align*}
One can see that, in this case, $\bar{c} = 3/16$ and if we want to sample up to error $\epsilon = 1/100 N$, Algorithm~\ref{alg:prop-sampling-qsp} samples correctly with $\Tilde{\bigO}(\sqrt{N})$ queries to the oracle.

\begin{theorem}
    No classical algorithm can solve Problem~\ref{def:prop-sampling} with less than $N - 1$ queries to the oracle.
\end{theorem}
\begin{proof}
    Let $P_c(\cdot) = \frac{c(x)}{\sum_y c(y)}$ be the probability distribution to approximate for instance $c$. Assume a classical algorithm $\mathcal{A}$ doing at most $N - 2$ queries to the oracle, and let $x_1, x_2$ be two of the values not queried by $\mathcal{A}$. If we take any instance $c'$ by only modifying $c(x_1), c(x_2)$ (in such a way that $\bar{c}' = \bar{c}$), then $\mathcal{A}$ will return the same probability distribution also for these two values, call it $P_{\mathcal{A}}(\cdot)$. Let us consider two separate cases: if one of $x_1, x_2 \in A$ (w.l.o.g. $x_1$), then $P_c(x_1) = 4/3N$, and the intervals of admitted values for correctness are at most $2/100 N$ long. Therefore, if we choose $c'(x_1) = 0$ (and $c'(x_2)$ so that $\bar{c}' = \bar{c}$), then $P_{c'}(x_1) = 0$, and $P_\mathcal{A}(x_1)$ cannot be $\epsilon$-close to both these probabilities.

    If $x_1, x_2 \in B$ setting $c'(x_1) = 0, c'(x_2) = \frac{1}{4}$ gives $P_{c'}(x_1) = 0 < P_{c}(x_1) - \frac{2}{100 N}$, therefore even in this case, $P_{\mathcal{A}}(x_1)$ cannot well-approximate both $c, c'$. We conclude that $\mathcal{A}$ is not correct.
\end{proof}
\section{Faster and Faster Approximations}
\label{sec:inductive-smoothening}
In this section we would like to improve on the rate of convergence of Theorem~\ref{thm:phase-extraction-jackson-rate}, which we can plug into Algorithm~\ref{alg:prop-sampling-qsp} to achieve further speed-up. We restate an important result about Fourier series here, which we also used in the proof of Theorem~\ref{thm:phase-extraction-jackson-rate}:
\begin{theorem}[Jackson~\cite{jacksonTheoryApproximation1930a}, Corollary 3, p.\ 22]
    \label{thm:jackson-rate}
    If $f : \R \rightarrow \R$ is a $2\pi$-periodic function such that its $p$-th derivative is $K$-Lipschitz continuous, then its Fourier sum $S_d(x)$ of $d$-th degree satisfies
    \begin{align*}
        ||S_d - f||_\R \le K \frac{A_p \log d}{d^{p+1}}
    \end{align*}
    where $A_p$ is a constant depending only on $p$.
\end{theorem}
In our example we took $p = 1$, we smoothed out $\phi_\delta$ so that its first derivative is continuous and we noticed that it is $\frac{2}{\delta^2}$-Lipschitz. This led to $d = \Tilde\bigO\left(\frac{1}{\delta} \sqrt{\frac{1}{\epsilon}}\right)$ to bound the approximation error by $\epsilon$ on the unit circle. However in the application for proportional sampling we could do more: we could smooth out all the derivatives up to some constant $p$, taking advantage of the fact that $\delta = \pi/2$.

\begin{figure*}
    \centering
    \def\dlt{0.4}
\begin{tikzpicture}
\begin{axis}[
    width=250pt,height=110pt,
    xmin=pi-0.5,xmax=pi+0.5,
    ymin=-1.2,ymax=1.45,
    samples=50,
    xtick={pi-\dlt, pi, pi+\dlt},
    xticklabels={$\pi - \delta$, $\pi$, $\pi + \delta$},
    ytick={-1, 0, 1},
    yticklabels={$-\frac{2}{\delta^2}$, $0$, $\frac{2}{\delta^2}$},
    grid style={line width=.1pt, draw=gray!10}]

    \addplot[red, ultra thick, domain=0:pi-\dlt] (x, 0);
    \addplot[red, ultra thick, domain=pi+\dlt:4] (x, 0);
    \addplot[red, ultra thick, domain=pi-\dlt:pi] (x, -1);
    \addplot[red, ultra thick, domain=pi:pi+\dlt] (x, 1);

    \draw [dashed] (axis cs:{pi-\dlt},-2) -- (axis cs:{pi-\dlt},2);
    \draw [dashed] (axis cs:{pi},-2) -- (axis cs:{pi},2);
    \draw [dashed] (axis cs:{pi+\dlt},-2) -- (axis cs:{pi+\dlt},2);
    
    \node at (axis cs:2.9,1.1) {$g^{(2)}_1 \equiv \phi''_\delta$};

    \addplot[mark=*] coordinates {(pi-\dlt,0)};
    \addplot[mark=o] coordinates {(pi-\dlt,-1)};
    \addplot[mark=*] coordinates {(pi,-1)};
    \addplot[mark=o] coordinates {(pi,1)};
    \addplot[mark=*] coordinates {(pi+\dlt,1)};
    \addplot[mark=o] coordinates {(pi+\dlt,0)};
\end{axis}
\end{tikzpicture}
\undef\dlt
    \def\dlt{0.4}
\begin{tikzpicture}
\begin{axis}[
    width=250pt,height=110pt,
    xmin=pi-0.5,xmax=pi+0.5,
    ymin=-1.2,ymax=1.45,
    samples=50,
    xtick={pi-\dlt, pi, pi+\dlt},
    xticklabels={$\pi - \delta$, $\pi$, $\pi + \delta$},
    ytick={-1, 0, 1},
    yticklabels={$-\frac{4}{\delta^2}$, $0$, $\frac{4}{\delta^2}$},
    grid style={line width=.1pt, draw=gray!10}]

    \addplot[blue, ultra thick] coordinates {
        (0, 0) (pi-\dlt, 0) (pi-\dlt/2, -1) (pi+\dlt/2, 1) (pi+\dlt, 0) (4,0)
    };

    \draw [dashed] (axis cs:{pi-\dlt},-2) -- (axis cs:{pi-\dlt},2);
    \draw [dashed] (axis cs:{pi},-2) -- (axis cs:{pi},2);
    \draw [dashed] (axis cs:{pi+\dlt},-2) -- (axis cs:{pi+\dlt},2);
    
    \node at (axis cs:2.9,1.1) {$g^{(2)}_2$};
\end{axis}
\end{tikzpicture}
\undef\dlt
    \def\dlt{0.4}
\begin{tikzpicture}
\begin{axis}[
    width=250pt,height=110pt,
    xmin=pi-0.5,xmax=pi+0.5,
    ymin=-1.2,ymax=1.45,
    samples=50,
    xtick={pi-\dlt, pi, pi+\dlt},
    xticklabels={$\pi - \delta$, $\pi$, $\pi + \delta$},
    ytick={-1, 0, 1},
    yticklabels={$-\frac{8}{\delta^3}$, $0$, $\frac{8}{\delta^3}$},
    grid style={line width=.1pt, draw=gray!10}]

    \addplot[blue, ultra thick, domain=0:pi-\dlt] (x, 0);
    \addplot[blue, ultra thick, domain=pi-\dlt:pi-\dlt/2] (x, -1);
    \addplot[blue, ultra thick, domain=pi-\dlt/2:pi+\dlt/2] (x, 1);
    \addplot[blue, ultra thick, domain=pi+\dlt/2:pi+\dlt] (x, -1);
    \addplot[blue, ultra thick, domain=pi+\dlt:4] (x, 0);

    \draw [dashed] (axis cs:{pi-\dlt},-2) -- (axis cs:{pi-\dlt},2);
    \draw [dashed] (axis cs:{pi-\dlt/2},-2) -- (axis cs:{pi-\dlt/2},2);
    \draw [dashed] (axis cs:{pi},-2) -- (axis cs:{pi},2);
    \draw [dashed] (axis cs:{pi+\dlt/2},-2) -- (axis cs:{pi+\dlt/2},2);
    \draw [dashed] (axis cs:{pi+\dlt},-2) -- (axis cs:{pi+\dlt},2);
    
    \node at (axis cs:2.9,1.1) {$g^{(3)}_2$};

    \addplot[mark=*] coordinates {(pi-\dlt,0)};
    \addplot[mark=o] coordinates {(pi-\dlt,-1)};
    \addplot[mark=*] coordinates {(pi-\dlt/2,-1)};
    \addplot[mark=o] coordinates {(pi-\dlt/2,1)};
    \addplot[mark=*] coordinates {(pi+\dlt/2,1)};
    \addplot[mark=o] coordinates {(pi+\dlt/2,-1)};
    \addplot[mark=*] coordinates {(pi+\dlt,-1)};
    \addplot[mark=o] coordinates {(pi+\dlt,0)};
\end{axis}
\end{tikzpicture}
\undef\dlt
    \def\dlt{0.4}
\begin{tikzpicture}
\begin{axis}[
    width=250pt,height=110pt,
    xmin=pi-0.5,xmax=pi+0.5,
    ymin=-1.2,ymax=1.45,
    samples=50,
    xtick={pi-\dlt, pi, pi+\dlt},
    xticklabels={$\pi - \delta$, $\pi$, $\pi + \delta$},
    ytick={-1, 0, 1},
    yticklabels={$-\frac{16}{\delta^3}$, $0$, $\frac{16}{\delta^3}$},
    grid style={line width=.1pt, draw=gray!10}]

    \addplot[green, ultra thick] coordinates {
        (0, 0) (pi-\dlt, 0) (pi-3*\dlt/4, -1) (pi-\dlt/4, 1) (pi, 0) (pi+\dlt/4, 1) (pi+3*\dlt/4, -1) (pi+\dlt, 0) (4,0)
    };

    \draw [dashed] (axis cs:{pi-\dlt},-2) -- (axis cs:{pi-\dlt},2);
    \draw [dashed] (axis cs:{pi-\dlt/2},-2) -- (axis cs:{pi-\dlt/2},2);
    \draw [dashed] (axis cs:{pi},-2) -- (axis cs:{pi},2);
    \draw [dashed] (axis cs:{pi+\dlt/2},-2) -- (axis cs:{pi+\dlt/2},2);
    \draw [dashed] (axis cs:{pi+\dlt},-2) -- (axis cs:{pi+\dlt},2);
    
    \node at (axis cs:2.9,1.1) {$g^{(3)}_3$};
\end{axis}
\end{tikzpicture}
\undef\dlt
    \def\dlt{0.4}
\begin{tikzpicture}
\begin{axis}[
    width=250pt,height=110pt,
    xmin=pi-0.5,xmax=pi+0.5,
    ymin=-1.2,ymax=1.45,
    samples=50,
    xtick={pi-\dlt, pi, pi+\dlt},
    xticklabels={$\pi - \delta$, $\pi$, $\pi + \delta$},
    ytick={-1, 0, 1},
    yticklabels={$-\frac{64}{\delta^4}$, $0$, $\frac{64}{\delta^4}$},
    grid style={line width=.1pt, draw=gray!10}]

    \addplot[green, ultra thick, domain=0:pi-\dlt] (x, 0);
    \addplot[green, ultra thick, domain=pi-\dlt:pi-3*\dlt/4] (x, -1);
    \addplot[green, ultra thick, domain=pi-3*\dlt/4:pi-\dlt/2] (x, 1);
    \addplot[green, ultra thick, domain=pi-\dlt/2:pi-\dlt/4] (x, 1);
    \addplot[green, ultra thick, domain=pi-\dlt/4:pi] (x, -1);
    \addplot[green, ultra thick, domain=pi:pi+\dlt/4] (x, 1);
    \addplot[green, ultra thick, domain=pi+\dlt/4:pi+\dlt/2] (x, -1);
    \addplot[green, ultra thick, domain=pi+\dlt/2:pi+3*\dlt/4] (x, -1);
    \addplot[green, ultra thick, domain=pi+3*\dlt/4:pi+\dlt] (x, 1);
    \addplot[green, ultra thick, domain=pi+\dlt:4] (x, 0);

    \draw [dashed] (axis cs:{pi-\dlt},-2) -- (axis cs:{pi-\dlt},2);
    \draw [dashed] (axis cs:{pi-3*\dlt/4},-2) -- (axis cs:{pi-3*\dlt/4},2);
    \draw [dashed] (axis cs:{pi-\dlt/2},-2) -- (axis cs:{pi-\dlt/2},2);
    \draw [dashed] (axis cs:{pi-\dlt/4},-2) -- (axis cs:{pi-\dlt/4},2);
    \draw [dashed] (axis cs:{pi},-2) -- (axis cs:{pi},2);
    \draw [dashed] (axis cs:{pi+\dlt/4},-2) -- (axis cs:{pi+\dlt/4},2);
    \draw [dashed] (axis cs:{pi+\dlt/2},-2) -- (axis cs:{pi+\dlt/2},2);
    \draw [dashed] (axis cs:{pi+3*\dlt/4},-2) -- (axis cs:{pi+3*\dlt/4},2);
    \draw [dashed] (axis cs:{pi+\dlt},-2) -- (axis cs:{pi+\dlt},2);
    
    \node at (axis cs:2.8,1.1) {$g^{(4)}_3$};

    \addplot[mark=*] coordinates {(pi-\dlt,0)};
    \addplot[mark=o] coordinates {(pi-\dlt,-1)};
    \addplot[mark=*] coordinates {(pi-3*\dlt/4,-1)};
    \addplot[mark=o] coordinates {(pi-3*\dlt/4,1)};
    \addplot[mark=*] coordinates {(pi-\dlt/4,1)};
    \addplot[mark=o] coordinates {(pi-\dlt/4,-1)};
    \addplot[mark=*] coordinates {(pi,-1)};
    \addplot[mark=o] coordinates {(pi,1)};
    \addplot[mark=*] coordinates {(pi+\dlt/4,1)};
    \addplot[mark=o] coordinates {(pi+\dlt/4,-1)};
    \addplot[mark=*] coordinates {(pi+3*\dlt/4,-1)};
    \addplot[mark=o] coordinates {(pi+3*\dlt/4,1)};
    \addplot[mark=*] coordinates {(pi+\dlt,1)};
    \addplot[mark=o] coordinates {(pi+\dlt,0)};
\end{axis}
\end{tikzpicture}
\undef\dlt
    \def\dlt{0.4}
\begin{tikzpicture}
\begin{axis}[
    width=250pt,height=110pt,
    xmin=pi-0.5,xmax=pi+0.5,
    ymin=-1.2,ymax=1.45,
    samples=50,
    xtick={pi-\dlt, pi, pi+\dlt},
    xticklabels={$\pi - \delta$, $\pi$, $\pi + \delta$},
    ytick={-1, 0, 1},
    yticklabels={$-\frac{128}{\delta^4}$, $0$, $\frac{128}{\delta^4}$},
    grid style={line width=.1pt, draw=gray!10}]

    \addplot[yellow, ultra thick] coordinates {
        (0, 0) (pi-\dlt, 0) (pi-7*\dlt/8, -1) (pi-5*\dlt/8, 1) (pi-\dlt/2, 0) (pi-3*\dlt/8, 1) (pi-\dlt/8, -1) (pi+\dlt/8, 1) (pi+3*\dlt/8, -1) (pi+\dlt/2, 0) (pi+5*\dlt/8, -1) (pi+7*\dlt/8, 1) (pi+\dlt, 0) (4,0)
    };

    \draw [dashed] (axis cs:{pi-\dlt},-2) -- (axis cs:{pi-\dlt},2);
    \draw [dashed] (axis cs:{pi-3*\dlt/4},-2) -- (axis cs:{pi-3*\dlt/4},2);
    \draw [dashed] (axis cs:{pi-\dlt/2},-2) -- (axis cs:{pi-\dlt/2},2);
    \draw [dashed] (axis cs:{pi-\dlt/4},-2) -- (axis cs:{pi-\dlt/4},2);
    \draw [dashed] (axis cs:{pi},-2) -- (axis cs:{pi},2);
    \draw [dashed] (axis cs:{pi+\dlt/4},-2) -- (axis cs:{pi+\dlt/4},2);
    \draw [dashed] (axis cs:{pi+\dlt/2},-2) -- (axis cs:{pi+\dlt/2},2);
    \draw [dashed] (axis cs:{pi+3*\dlt/4},-2) -- (axis cs:{pi+3*\dlt/4},2);
    \draw [dashed] (axis cs:{pi+\dlt},-2) -- (axis cs:{pi+\dlt},2);
    
    \node at (axis cs:2.8,1.1) {$g^{(4)}_4$};
\end{axis}
\end{tikzpicture}
\undef\dlt
    \caption{Construction of the fourth derivative of $g_4 \in C^4$. To construct the derivative of the next function, we linearize the `last' derivative in order to make it continuous, and we do so by replacing the rectangles with triangles of the same area, so that the integral over $I = (\pi - \delta, \pi + \delta)$ is preserved. In order to obtain $g_4$, we will integrate four times, keeping in mind that every derivative has value $0$ at the origin, except for the first one, which has value $1/\pi$.}
    \label{fig:function-smoothening-1}
\end{figure*}
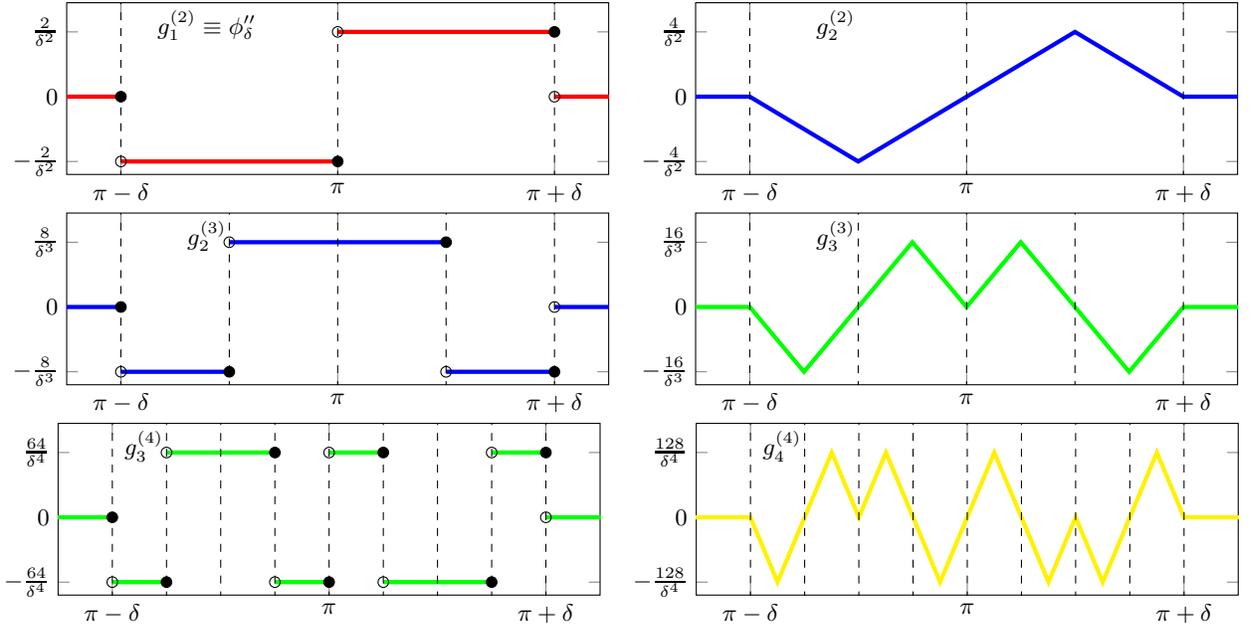

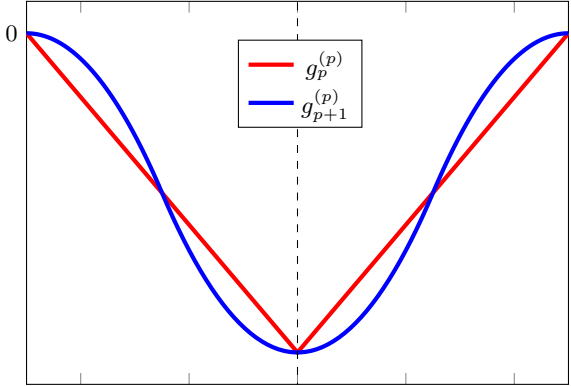
\begin{figure}
    \centering
    \def\dlt{0.4}
\begin{tikzpicture}
\begin{axis}[
    width=250pt,height=190pt,
    xmin=-0.5,xmax=+0.5,
    ymin=-0.55,ymax=0.05,
    samples=50,
    xtick={},
    xticklabels={},
    ytick={0},
    yticklabels={$0$},
    grid style={line width=.1pt, draw=gray!10},
    legend style={at={(0.39,0.9)}, anchor=north west}]

    \addplot[red, ultra thick] coordinates {
        (-0.5, 0) (0, -0.5) (0.5,0)
    };

    \addplot[blue, ultra thick, domain=-0.5:-0.25] (x, {-4*x*x-4*x-1});
    \legend{$g^{(p)}_p$, $g^{(p)}_{p+1}$}
    \addplot[blue, ultra thick, domain=-0.25:0.25] (x, {4*x*x-0.5});
    \addplot[blue, ultra thick, domain=0.25:0.5] (x, {-4*x*x+4*x-1});

    \draw [dashed] (axis cs:{0},-2) -- (axis cs:{0},2);
    
    \node at (axis cs:2.9,1.1) {$g^{(3)}_3$};
\end{axis}
\end{tikzpicture}
\undef\dlt
    \caption{Comparison between a triangle and its first-order smoothing. If this is the $j$-th triangle of $g^{(p)}_{p+1}$, then the interval of the plot is $I^{p-1}_j$, which is split into the two intervals $I^p_{2j}, I^p_{2j+1}$. One can see that the difference of the two functions is odd in each of the two sub-intervals, and even in $I^{p-1}_j$.}
    \label{fig:function-smoothening-parabola}
\end{figure}

We generalize the smoothing procedure shown in Section~\ref{sec:phase-extraction-problem}: one can see that $\phi'_\delta$ is a piecewise linear function.
\begin{claim}
    For $p \ge 1$, we can construct $g_p \in C^p$ that is $2\pi$-periodic, and
    \begin{align}
        \label{eq:general-construction-identity-condition}
        g_p(x) = \frac{x}{\pi}
    \end{align}
    for every $x \in (-\pi + \delta, \pi - \delta)$.
\end{claim}

\begin{proof}
    Note that $g_1 \equiv \phi_\delta$ gives a base for our inductive construction. First of all, Eq.~(\ref{eq:general-construction-identity-condition}) holds iff the derivatives of $g_p$ satisfy:
    \begin{align}
        \label{eq:smoothening-derivative-values}
        g^{(k)}_p(x) & = 
        \begin{cases}
            \frac{1}{\pi} & k = 1 \\
            0 & k \neq 1
        \end{cases}
    \end{align}
    for $x \in (-\pi + \delta, \pi - \delta)$ and every $k \ge 0$. By this condition, we are only allowed to change the behaviour of the function outside of this domain, thus from now on we will only consider the interval $I = (\pi-\delta, \pi+\delta)$.
    
    Starting from $g^{(1)}_1$, which is continuous and piecewise linear, its derivative, $g^{(2)}_1$ will be piecewise constant, i.e., its integral will be given by a sequence of rectangles (Figure~\ref{fig:function-smoothening-1}, top left). We construct the second derivative of the next function $g^{(2)}_2$ by replacing the rectangles of $g^{(1)}_2$ by triangles of the same area (Figure~\ref{fig:function-smoothening-1}, top right). Then $g^{(3)}_2$ will be again piecewise constant, where the number of intervals in which this function is subdivided is twice the intervals of our starting function $g^{(2)}_1$. In general, $g^{(p)}_{p+1}$ will be piecewise constant, $g^{(p+1)}_{p+1}$ is obtained with the above procedure, and its pieces divide $I$ into $2^p$ equal segments. At this point, $g_p$ is easily defined by integrating the constructed derivatives $p+1$ times, with boundary conditions given by Eq.~(\ref{eq:smoothening-derivative-values}):
    \begin{align*}
        g_{p+1}(x) = \int_0^x \left( \frac{1}{\pi} + \underbrace{\int_0^x \cdots \int_0^x}_{p} g^{(p+1)}_{p+1}(x) \ dx^p \right) dx\ .
    \end{align*}
    Figure~\ref{fig:function-smoothening-1} shows the construction of the fourth derivative of $g_4$. We now prove all the claimed properties: $g_p \in C^p$ by construction, and we never changed the derivatives within $\pm (\pi - \delta)$, thus Eq.~(\ref{eq:smoothening-derivative-values}) is preserved (and so is also (\ref{eq:general-construction-identity-condition})).
    We only need to prove periodicity: the key here is to notice that $g^{(2)}_p$ satisfies
    \begin{align*}
        \int_I g^{(2)}_p(t) \ dt = 0 \ ,
    \end{align*}
    implying that $g^{(1)}_p$ is $\frac{1}{\pi}$ at the boundaries of $I$. This because $g^{(2)}_p$ has the same form as $g^{(2)}_1$ in Figure~\ref{fig:function-smoothening-1}, except for the fact that, instead of having two opposite rectangles, we have two opposite (smoother) shapes which will cancel out with one another in the integral (in the case of $p = 2$, these shapes are the triangles as in Figure~\ref{fig:function-smoothening-1}, top right). This shows that $g^{(1)}_p$ is $2\pi$-periodic. It is now sufficient to prove
    \begin{align}
        \label{eq:smoothening-ultimate-integral-condition}
        g_p(2\pi) - g_p(0) = \int_0^{2\pi} g^{(1)}_p(t) \ dt = 0
    \end{align}
    in order to prove periodicity of $g_p$.
    This is already true for $p = 1$, by design: in this case the shape of $g^{(1)}_1$ in $I$ is one big triangle pointing downwards, whose area was chosen in order to satisfy (\ref{eq:smoothening-ultimate-integral-condition}).

    For $p > 1$, this triangle is replaced by a smoother shape, and it is sufficient to prove that such shape has the same area as the original triangle in order to preserve the above integral. Thus we prove the following by induction:
    \begin{align*}
        \int_I g^{(1)}_{p+1}(t) - g^{(1)}_p(t) \ dt = 0
    \end{align*}
    so that by telescoping and linearity of integral the claim will follow. We introduce some notation: for $0 \le j < 2^k$, $I^k_j$ is the $j$-th slice of $I$ (starting from its left endpoint $\pi - \delta$) after dividing it into $2^k$ sub-intervals. Notice that it holds that $I^{k-1}_j = I^k_{2j} \cup I^k_{2j+1}$, for $0 \le j < 2^{k-1}$. When we say that a function is even (odd) in some interval, we mean that the function restricted to that interval is symmetric (anti-symmetric) with respect to the middle point.
    
    As a base for an inductive argument, one can see that the function $g^{(p)}_{p+1} - g^{(p)}_p$ is odd in $I^p_j$, for every $j$: this because $g^{(p)}_p$ is the side of a triangle in $I^p_{j}$, $g^{(p)}_{p+1}$ is the side of a parabolic bell and we can directly check that the difference is odd in $I^p_{j}$ (see Figure~\ref{fig:function-smoothening-parabola}).
    
    Assuming that $g^{(k)}_{p+1} - g^{(k)}_p$ is odd in $I^k_j$ for every $j$ and some $1 < k \le p$, we have that
    \begin{align*}
        \int_{L(I^k_j)}^{R(I^k_j)} g^{(k)}_{p+1}(t) - g^{(k)}_p(t) \ dt = 0 \ ,
    \end{align*}
    where $L(J), R(J)$ are, respectively, the left and right endpoints of an interval $J$. The function $g^{(k)}_p$ is a sequence of (smoothed) triangles in $I$, but this shape must be equal to $0$ at the endpoints: this because these shapes are symmetric and all equal up to sign by construction and the left endpoint of the leftmost shape must preserve continuity with the constant function outside of $I$, in both $g^{(k)}_p$ and its smoothed version $g^{(k)}_{p+1}$. Hence, at the base of the shapes the difference is $0$.
    One of $L(I^k_j), R(I^k_j)$, corresponds to one of the base points of the shape, while the other is the peak (which one depends on whether $I^k_j$ contains the first or second half of the shape, i.e., the parity of $j$). By oddness in this interval, also at the peak the difference is $0$. Considering the leftmost shape, the function $g^{(k)}_{p+1} - g^{(k)}_p$ is even in $I^{k-1}_0 = I^k_0 \cup I^k_1$. Therefore,
    \begin{align*}
        g^{(k-1)}_{p+1} - g^{(k-1)}_p & = \int_{0}^x g^{(k)}_{p+1}(t) - g^{(k)}_p(t) \ dt \\
        & = \int_{L(I^k_0)}^x g^{(k)}_{p+1}(t) - g^{(k)}_p(t) \ dt \\
        & = \int_{R(I^k_0)}^x g^{(k)}_{p+1}(t) - g^{(k)}_p(t) \ dt \ ,
    \end{align*}
    and the last expression is clearly odd in $I^{k-1}_0$, since $R(I^k_0)$ is the middle point.
    
    By a simple induction on $j$ the result can be extended to every shape, implying that $g^{(k-1)}_{p+1} - g^{(k-1)}_p$ is odd in $I^{k-1}_j$ for every $j$, and in particular, $g^{(1)}_{p+1} - g^{(1)}_p$ is odd in $I^1_0, I^1_1$, and both integrals cancel out, giving zero also on their union, which is the whole interval $I$.
\end{proof}
All we need now is to bound the Lipschitz constant for $g^{(p)}_p$: we already know that $g^{(1)}_1$ is $K_1$-Lipschitz with
\begin{align*}
    K_1 = \frac{2}{\delta^2}
\end{align*}
since its derivative $g^{(2)}_1$ is bounded by this number in absolute value. If $K_p$ is the Lipschitz constant of $g^{(p)}_p$, then
$$K_p = \sup |g^{(p+1)}_p| \ .$$
In order to obtain the same integral when we transform the rectangles into triangles, we need twice the height, thus the triangles of $g^{(p+1)}_{p+1}$ are as high as $2 K_p$. The slopes of these triangles are then $2 K_p$ over the length of a single segment which is $\frac{2 \delta}{2^{p+1}}$. Therefore, the Lipschitz constant for $g^{(p+1)}_{p+1}$ is
\begin{align*}
    K_{p+1} = \frac{2 K_p 2^{p+1}}{2 \delta} = \frac{K_p 2^{p+1}}{\delta} \ ,
\end{align*}
and this recurrence relation has unique solution
\begin{align}
    \label{eq:lipschitz-constant-recurrence-solution}
    K_p = \frac{\sqrt{2}^{p(p+1)}}{\delta^{p+1}} \ .
\end{align}
In conclusion, we can now use Theorem~\ref{thm:jackson-rate} to extend Theorem~\ref{thm:phase-extraction-jackson-rate} with a similar argument.
\begin{theorem}
    \label{thm:phase-extraction-jackson-rate-enhanced}
    Let $p > 0$ be a fixed constant. The function $\bar{g}_p : U(1) \rightarrow \R$ defined as
    \begin{align*}
        \bar{g}_p(e^{ix}) = g_p(x)
    \end{align*}
    for every $x \in \R$ can be $\epsilon$-approximated on the unit circle using a polynomial of degree
    \begin{align*}
        d = \Tilde{\bigO}\left(\frac{A_p \sqrt{2}^p}{\delta} \left(\frac{1}{\epsilon}\right)^{\frac{1}{p+1}} \right).
    \end{align*}
\end{theorem}
As a Corollary, for a constant $p$ Algorithm~\ref{alg:prop-sampling-qsp} can be done in only
\begin{align*}
    \Tilde{\bigO}\left( \frac{1}{\bar{c}^{\frac{1}{2} + \frac{1}{p}}} \sqrt[p]{\frac{1}{\epsilon}} \right)
\end{align*}
total oracle queries.

\section{Conclusions}
We consider the phase-extraction problem, and we showed that, given a unitary $U = e^{i\pi H}$ and its inverse $U^{\dag}$, we could implement a block-encoding of $\phi(H)$ for some smooth function $\phi(x)$. The word `smooth' here means existence and continuity of the derivatives: the higher the number of continuous derivatives that a function has, the faster its Fourier sum (and thus the Laurent polynomial on the eigenphases) uniformly converges to that function. We are confident this can have many more applications beyond what is shown in this work. It is also worth remarking that Jackson showed that the convergence rate of a Fourier series is almost-optimal, in the sense that no trigonometric (or, equivalently, complex exponential) series can approximate the desired function faster, up to that $\log d$ factor~\cite[p.\ 21]{jacksonTheoryApproximation1930a}. Also remember that `smoothing' a function, i.e., replacing its derivative with a continuous function, does not give faster convergence for free in general, as its derivative will become steep in the points where we smooth out discontinuities, and this translates to a high Lipschitz constant: a~clear example is given by Eq.~\ref{eq:lipschitz-constant-recurrence-solution}, but in that case, fortunately, nothing depends on the size of the input $N$, and thus does not influence the asymptotic query complexity of Algorithm~\ref{alg:prop-sampling-qsp}, although the constant factor can become large even for $p = 20$. From a theoretical point of view, this work shows that, for any $\eta > 0$, there is an algorithm with query complexity 
$$\Tilde{\bigO}\left(\frac{1}{\bar{c}^{\frac{1}{2} + \eta}} \frac{1}{\epsilon^\eta} \right)$$
solving the proportional-sampling problem. This statement seems to suggest there exists an algorithm which directly solves the problem with $\eta = 0$, and an open question would be to find such algorithm.

It is also interesting to remark that Theorems~\ref{thm:haah-construction},~\ref{thm:haah-completion} indeed allow the construction for any $\phi$, even complex-valued, provided that its absolute value is reciprocal.

One could think that, in Section~\ref{sec:prop-sampling}, instead of using the linear function in the phase-extraction subroutine, we could approximate the square root and then apply the transformation directly on $e^{i \pi c(x)}$. However, in the case of proportional sampling this would be inconvenient, as the derivative of the square root function has a discontinuity with an infinite jump around 0, and we could not choose a constant $\delta$ if we had values of the oracle that are too close to $0$.

\section*{Acknowledgements}
I would like to thank William Schober and Stefan Wolf for insightful feedback and discussions. This work was supported by the Swiss National Science Foundation (SNF), grant No. \texttt{200020\_182452}.

\appendix

\section{Haah's Construction of QSP Polynomials}
\label{apx:haah-construction}

Here we briefly review the Quantum Signal Processing technique, in its Laurent-polynomial formulation. We suggest~\cite{martynGrandUnificationQuantum2021, haahProductDecompositionPeriodic2019, chaoFindingAnglesQuantum2020} for a more comprehensive discussion on the topic. Given an arbitrary unitary $W$, consider the following matrix:
\begin{align*}
    \Tilde{W} = \begin{bmatrix}
        W & 0 \\
        0 & W^\dag
    \end{bmatrix} .
\end{align*}
This corresponds to a circuit that applies $W$ or its inverse $W^\dag$, depending on the state of a control qubit. If we feed an eigenstate $\ket{\theta}$ with associated eigenvalue $w = e^{i\theta}$, by a phase kickback the control qubit undergoes a unitary transformation of the form
\begin{align*}
    \Tilde{w} = \begin{bmatrix}
        w & 0 \\ 0 & w^{-1}
    \end{bmatrix} .
\end{align*}
This is the idea behind a technique called \emph{qubitization}~\cite{lowHamiltonianSimulationQubitization2019}, where one can apply Quantum Signal Processing polynomials simultaneously on all the eigenvalues of a unitary given as a black-box (in our case $W$), giving birth to the so-called \emph{quantum eigenvalue transform}. From now on we consider $w$ as an arbitrary unitary eigenvalue, and we want to understand which transformations can be realized. One can see that a unitary of the form
\begin{align}
    \label{eq:haah-primitive-decomposition}
    U(\Tilde{w}) = Q_n \cdot \Tilde{w} \cdot Q_n^\dag \cdots Q_1 \cdot \Tilde{w} \cdot Q_1^\dag \cdot Q_0
\end{align}
gives a $2$-by-$2$ matrix of \emph{Laurent polynomials} (or equivalently, a Laurent polynomial with $2$-by-$2$ matrices as coefficients), and the maximum degree of this polynomial is $n$. Roughly speaking, we construct a circuit where the $Q_i$'s act on the control qubit and $\Tilde{w}$ is replaced with the controlled unitary $\Tilde{W}$. We are mainly interested in two sub-algebras of the matrices above.
\begin{definition}
    The \emph{Haah algebra} $H$ is the sub-algebra of the ring of Laurent polynomials over $2$-by-$2$ complex matrices of the form
    \begin{align*}
        a(\Tilde{w}) + b(\Tilde{w}) \cdot iX + c(\Tilde{w}) \cdot iY + d(\Tilde{w}) \cdot iZ
    \end{align*}
    where $a, b, c, d$ are polynomials with real coefficients.
\end{definition}
\noindent In other words, the Haah algebra is the polynomial ring $H \simeq \R[iX, iZ, \Tilde{w}]$. An important sub-algebra of the Haah algebra, which is also the most interesting one for practical implementations is the \emph{Low algebra}.
\begin{definition}
    The \emph{Low algebra} $L$ is the sub-algebra of the ring of Laurent polynomials over $2$-by-$2$ complex matrices of the form
    \begin{align*}
        a(\Tilde{w}) + b(\Tilde{w}) \cdot iX
    \end{align*}
    where $a, b$ are Laurent polynomials with real coefficients.
\end{definition}
\noindent Analogously as before we have $L \simeq \R[iX, \Tilde{w}]$.

\begin{theorem}[Haah~\cite{haahProductDecompositionPeriodic2019}]
    \label{thm:haah-construction}
    An element $U(w) \in H$ that is unitary and has definite parity (i.e.\ $a, b, c, d$ are either all even or all odd Laurent polynomials) is always decomposable as in (\ref{eq:haah-primitive-decomposition}), in a unique way (up to global phase).

    Furthermore, if $U(w) \in L$ then all $Q_i \in L$. In particular, they will be $X$-rotations.
\end{theorem}
\noindent The Low algebra is useful exactly because it is practical, as we only have to define the phases of the $X$-rotations. The decomposition using the Low algebra thus simply reduces to the $W_z$-convention of the traditional formulation of the QSP~\cite{lowMethodologyResonantEquiangular2016,martynGrandUnificationQuantum2021}. On the other hand, while Low algebra only allows the implementation of Laurent polynomials with real coefficients, Haah algebra can even represent polynomials with complex coefficients.

In applications, we are usually interested in a single transformation of the unitary, i.e., given $\Tilde{W}$ we would like to implement $f(W)$ for some Laurent polynomial $f$ that satisfies $|f|^2 \le 1$ on the unit circle. In order to do this we need an element $F(\Tilde{w})$ of the Haah algebra containing $f$ as a particular entry, which determines the initial state and the state to post-select on the control qubit. Here we will only use the top-left entry, $f(w) = \bra{0} F(\Tilde{w}) \ket{0}$, i.e.\ we will start and post-select $\ket{0}$. The question now is whether we can find the polynomials $a, b, c, d$ such that we obtain $f$ in the top-left entry.

\begin{theorem}[Haah's completion~\cite{haahProductDecompositionPeriodic2019}, restated]
    \label{thm:haah-completion}
    Given real-valued polynomials $a(w), d(w) : U(1) \rightarrow \R$, let 
    $$f(w) = a(w) + i d(w)$$
    be a function such that the polynomial $|f|^2 = a^2 + d^2$ is reciprocal (i.e.\ $|f(w)|^2 = |f(w^{-1})|^2$), has real coefficients and satisfies $|f(w)|^2 < 1$ on the unit circle.
    There exist polynomials $b(w), c(w)$ such that
    \begin{align*}
        a(\Tilde{w}) + b(\Tilde{w}) \cdot iX + c(\Tilde{w}) \cdot iY + d(\Tilde{w}) \cdot iZ \in SU(2)
    \end{align*}
    for every $t \in U(1)$.
\end{theorem}
Note that, with the given construction, $f(w)$ is present on the top-left corner of the matrix, while the bottom-right entry contains $(f(w^{-1}))^*$.
\begin{proof}
    On the unit circle, this translates to the condition
    \begin{align*}
        a^2(w) + b^2(w) + c^2(w) + d^2(w) \equiv 1.
    \end{align*}
    The polynomial $p(w) = 1 - a^2(w) - d^2(w)$ has degree $n' \le 2n$ (in Laurent polynomials leading terms could cancel out). Moreover, by the conditions on $|f|^2$, it has real coefficients, all its roots are outside the unit circle, and by reciprocity for any root $r$ also $1/r$ is a root. Therefore, using $\mathcal{D}$ to denote the multiset of roots within the unit circle, the expression becomes
    \begin{align}
        \label{eq:haah-construction-roots}
        p(t) = \alpha \prod_{r \in \mathcal{D}} (w - r) (1/w - r)
    \end{align}
    for some proportionality constant $\alpha$. Keep in mind that $|\mathcal{D}| = n'$. Let us now define
    \begin{align*}
        e(w) := w^{-\lfloor n'/2 \rfloor} \prod_{r \in \mathcal{D}} (w - r) \ .
    \end{align*}
    In this way, $\alpha \cdot e(w) \cdot e(1/w) = p(w)$. Note that the factor in front of the product is used to ‘center' the exponents of $e(w)$, so that its degree is $\lceil n'/2 \rceil \le n$. This however, does not affect the equality we just stated.
    Plugging $t = 1$ in Eq.\ (\ref{eq:haah-construction-roots}) gives real and positive expressions on both sides (complex roots cannot not give negative contributions, since they all come in complex conjugate pairs, by reality of the coefficients). Hence we conclude that $\alpha$ is positive.
    \begin{align*}
        p(t) & = \alpha \cdot e(w) \cdot e(1/w) \\
        & = \left( \frac{e(w) + e(1/w)}{2} \sqrt{\alpha} \right)^2 + \left( \frac{e(w) - e(1/w)}{2i} \sqrt{\alpha} \right)^2 \ ,
    \end{align*}
    and we can choose the two expressions in the tuples as $b(w), c(w)$, for example. They will have both degree $\le \lceil n'/2 \rceil \le n$.
\end{proof}
It is interesting to remark that the reciprocity of $|f(w)|^2$ implies the reciprocity of $a^2(w) + d^2(w)$. Under this assumption, we have that $a^2$ is reciprocal if and only if $d^2$ is reciprocal. If the square of a polynomial is reciprocal with real coefficients, then the original polynomial is either reciprocal itself with real coefficients (cosine transform), or anti-reciprocal with imaginary coefficients (sine transform). On the other hand, if $a^2$ is not reciprocal, then in order to maintain reciprocity of the sum we need $a^2(w) = d^2(1/w)$, or simply $a(w) = \pm d(1/w)$. In this last case, $a(w)$ and $d(w)$ do not need to have real nor imaginary coefficients as long as they are real-valued for any $w \in U(1)$. Also as pointed out by Haah, (anti-)reciprocity is not a severe restriction, since any function can be decomposed into a sum of reciprocal and anti-reciprocal components, which can be implemented separately and then summed using a circuit like in Figure~\ref{fig:block-encoding-sum-circuit}, incurring on at most a constant factor on the failure probability.
\section{Semi-oblivious amplitude amplification}
\label{apx:amplitude-amplification}

\noindent In Section~\ref{sec:phase-extraction-problem} we have the following unitary:
\begin{align*}
    S \ket{00} \ket{\psi}^{\otimes n} = \ket{00} H_c \ket{\psi} + \ket{\Phi}
\end{align*}
\smallskip

\noindent where $\braket{00}{\Phi} = 0$, and $H_c \ket{x} = \frac{1}{2} \sqrt{\Tilde{c}(x)} \ket{x}$. Unfortunately, we cannot apply an oblivious amplitude amplification~\cite{berryExponentialImprovementPrecision2014a} procedure, not even in its robust version~\cite{berrySimulatingHamiltonianDynamics2015}. This because we would need $H_c$ to be (almost-)unitary, but it can even have zero eigenvalues. On the other hand, however, we only need to transform one particular state, i.e.\ $\ket{+}^{\otimes n}$. We use the Quantum Singular Value Transform (QSVT)~\cite{gilyenQuantumSingularValue2019a, gilyenQuantumSingularValue2019c} in a similar way it was used in~\cite{martynGrandUnificationQuantum2021} for Grover's search. We introduce some notation: let $\ket{\Psi} = \ket{00} \ket{+}^{\otimes n}$ denote our initial state, while $\ket{w} = \ket{00} \frac{H_c \ket{+}^{\otimes n}}{||H_c \ket{+}^{\otimes n}||}$ is our (normalized) target state. If $\Pi' = \ketbra{00}{00} \otimes \id, \Pi = \ketbra{\Psi}{\Psi}$. Then
\begin{align*}
    \Pi' S \Pi \ket{\Psi} = \ket{00} H_c \ket{+}^{\otimes n} = \sigma \ket{w}
\end{align*}
where $\sigma = ||H_c \ket{+}^{\otimes n}|| = \Theta(\sqrt{\bar{c}})$. Since any state orthogonal to $\ket{\Psi}$ is canceled out by $\Pi$, then it means that a singular value decomposition is
\begin{align*}
    \Pi' S \Pi = \sigma \ketbra{w}{\Psi}.
\end{align*}
Notice that, in our case, reflections and rotations of $\Pi', \Pi$ can be easily implemented and we can use the polynomial that approximates $\erf(k[x - c])$ (see, e.g.\,~\cite[Section III]{martynGrandUnificationQuantum2021}), having constant success probability with only
$$\bigO\left(\frac{1}{\sigma}\right) = \bigO\left(\frac{1}{\sqrt{\bar{c}}}\right)$$
calls to $S$.

\bibliography{refs}

\end{document}